\documentclass[10pt]{article}
\usepackage{}
\usepackage{amssymb}
\usepackage{amsfonts}
\usepackage{bbm}
\usepackage{mathrsfs}
\usepackage{mathrsfs}
\usepackage{longtable}
\usepackage{supertabular}
\usepackage{amsfonts,mathrsfs,latexsym,amsmath,amssymb,amsthm}

\newtheorem{theorem}{Theorem}
\newtheorem{proposition}{Proposition}
\newtheorem{lemma}{Lemma}

\newtheorem{definition}{Definition}
\newtheorem{example}{Example}

\begin{document}
\title{\bf{Self-dual codes and quadratic residue codes over the ring $\mathbb{Z}_9+u\mathbb{Z}_9$}}
\author{{\bf  Jian Gao, XianFang Wang, Fang-Wei Fu}\\
 {\footnotesize \emph{Chern Institute of Mathematics and LPMC, Nankai University}}\\
  {\footnotesize  \emph{Tianjin, 300071, P. R. China}}}
\date{}

\maketitle \noindent {\small {\bf Abstract} In this paper, we introduce a new definitions of the Gray weight and the Gray map for linear codes over $\mathbb{Z}_9+u\mathbb{Z}_9$ with $u^2=u$. Some results on self-dual codes over this ring are investigated. Further, the structural properties of quadratic residue codes are also considered. Two self-dual codes with parameters $[22,11,5]$ and $[24,12,9]$ over $\mathbb{Z}_9$ are obtained.}
 \vskip 1mm



\vskip 3mm \baselineskip 0.2in

\section{Introduction}

Codes over finite rings have been studied since the early 1970s. There are a lot of works on codes over finite rings after the discovery that certain good nonlinear binary codes can be constructed from cyclic codes over $\mathbb{Z}_4$ via the Gray map \cite{Hammons}. Since then, many researchers have payed more and more attentions to study the codes over finite rings. In these studies, the group rings associated with codes are finite chain rings. Recently, Zhu et al. considered linear codes over the finite non-chain ring $\mathbb{F}_q+v\mathbb{F}_q$. In \cite{Zhu1}, they studied the cyclic codes over $\mathbb{F}_2+v\mathbb{F}_2$. It has shown that cyclic codes over this ring are principally generated. In the subsequent paper \cite{Zhu2}, they investigated a class of constacyclic codes over $\mathbb{F}_p+v\mathbb{F}_p$. In that paper, the authors proved that the image of a $(1-2v)$-constacyclic code of length $n$ over $\mathbb{F}_p+v\mathbb{F}_p$ under the Gray map is a cyclic code of length $2n$ over $\mathbb{F}_p$. Furthermore, they also asserted that $(1-2v)$-constacyclic codes over $\mathbb{F}_p+v\mathbb{F}_p$ are also principally generated.  More recently, Yildiz and Karadeniz \cite{Yildiz} studied the linear codes over the non-principal ring $\mathbb{Z}_4+u\mathbb{Z}_4$, where $u^2=0$. They introduced the MacWilliams identities for the complete, symmetrized and Lee weight enumerators. They also gave three methods to construct formally self-dual codes over $\mathbb{Z}_4+u\mathbb{Z}_4$.
\par
Self-dual codes are an important class of linear codes. They have connections to many fields of research such as lattices, designs and invariant \cite{Bannai,Cengellenmis}. The study of the connections with unimodular lattices has generated a lot of interests among the coding theory, such as coding theory on the rings $\mathbb{Z}_4$, $\mathbb{Z}_8$, $\mathbb{Z}_9$, $\mathbb{Z}_{2k}$ and $\mathbb{Z}_{2^m}$. In \cite{Balmaceda}, the mass formula for self-dual codes and a classification of self-dual codes of small lengths over $\mathbb{Z}_9$ were given.  They classified such codes of lengths up to $n = 8$ over the ring $\mathbb{Z}_9$. Self-dual codes of larger length over $\mathbb{Z}_9$ may be also interesting, and this is the main motivation to construct self-dual codes from quadratic residue codes in this paper.
\par
As a special class of cyclic codes, quadratic residue codes fall into the family of BCH codes and have proven to be a promising family of cyclic codes. They were first introduced by Gleason and since then have generated a lot of interests. This due to the fact that they enjoy good algebra properties and they contain source of good codes. Recently, quadratic residue codes over some special finite rings were introduced by many researchers \cite{Kaya1,Kaya2,Taeri}, and they also produced some good self-dual codes over finite fields \cite{Kaya1,Kaya2}.
\par
In this paper, we mainly study quadratic residue codes over the ring $R=\mathbb{Z}_9+u\mathbb{Z}_9$, where $u^2=u$. In Section 2, we define the Lee weight of the element of $R$, and introduce a Gray map. This map leads to some useful results on linear codes over $R$. In Section 3, we study some structural properties of quadratic residue codes over $R$. Further, some self-dual codes of larger length over $\mathbb{Z}_9$ are obtained from quadratic residue codes over $R$.
\section{Preliminaries }
Let $R=\mathbb{Z}_9+u\mathbb{Z}_9$, where $u^2=u$. Then $R$ is commutative and with characteristic $9$. Clearly, $R\simeq \mathbb{Z}_9[u]/(u^2-u)$. Moreover, $u$ and $1-u$ are primitive idempotents. Then any element $r$ of $R$ can be expressed uniquely as $r=ua+(1-u)b$. This expression is called the decomposition expression of $r$. The ring $R$ has the following properties: there are $9$ different ideals of $R$ and they are $(1)$, $(u+2)$, $(u+3)$, $(u-1)$, $(3)$, $(u)$, $(3u-3)$, $(3u)$, $(0)$; $R$ is a principal ring; $(u+2)$ and $(u+3)$ are the maximal ideals of $R$; $R$ is not a finite chain ring. Furthermore, for any element $r=ua+(1-u)b$ of $R$, $r$ is a unit if and only if $a\not\equiv 0 ({\rm mod}3)$ and $b\not\equiv 0 ({\rm mod}3)$.
\par
Let $A$ be an element of ${\rm GL}_2(\mathbb{Z}_9)$,  i.e. $A$ is an invertible matrix of order $2$ over $\mathbb{Z}_9$. For any element $r=ua+(1-u)b$, we define the following map
\begin{equation*}
\begin{split}
\Phi:~~R& \rightarrow \mathbb{Z}_9^2\\
 r=ua+(1-u)b&\mapsto (a,b)A.
 \end{split}
 \end{equation*}
For simplicity, we write $(a,b)A$ as $rA$. The map $\Phi$ defined above can also extended to $R^n$ as follows.
\begin{equation*}
\begin{split}
\Phi:~~R^n& \rightarrow \mathbb{Z}_9^{2n}\\
(c_0,c_1,\ldots,c_{n-1})&\mapsto (c_0A,c_1A,\ldots,c_{n-1}A).
 \end{split}
 \end{equation*}
\begin{definition}
The map $\Phi$ defined above is called the Gray map from $R^n$ to $\mathbb{Z}_9^{2n}$. Let $r=ua+(1-u)b$ be any element of $R$. Then the Lee weight of $r$ is defined as
   \begin{equation*}
   w_L(r)=w_H(rA).
   \end{equation*}
\end{definition}
\par
Define the Lee weight of a vector $\textbf{c}=(c_0, c_1, \ldots, c_{n-1})\in R^n$ to be the rational sum of the Lee weight of its components, i.e. $w_L(\textbf{c})=\sum_{i=0}^{n-1}w_L(c_i)$. For any elements $\textbf{c}_1, \textbf{c}_2 \in R^n$, the Lee distance is given by $d_L(\textbf{c}_1, \textbf{c}_2)=w_L(\textbf{c}_1- \textbf{c}_2)$.
\par
A code $\mathcal {C}$ of length $n$ over $R$ is a nonempty subset of $R^n$. $\mathcal {C}$ is a linear code if and only if $\mathcal {C}$ is an $R$-submodule of $R^n$. The minimum Lee distance  of $\mathcal {C}$ is the smallest nonzero Lee distance between all pairs of distinct codewords. The minimum Lee weight of $\mathcal {C}$ is the smallest nonzero Lee weight among all codewords. If $\mathcal {C}$ is a linear code, then the minimum Lee distance is the same as the minimum Lee weight.
\par

It is well known that the Hamming weight of the element $a$ of $\mathbb{Z}_9$ is defined as $$w_H(r)=1~{\rm if}~r\neq 0, ~{\rm or}~0~{\rm otherwise}.$$ Then, by the Definition 1, we have the following result directly.

\begin{proposition}
The Gray map $\Phi$ is a distance-preserving map from $R^n$ (Lee distance) to $\mathbb{Z}_9^{2n}$ (Hamming distance) and it is also $\mathbb{Z}_9$-linear.
\end{proposition}

Let $\mathcal {C}$ be a $(n, M, d_L)$ linear code over $R$, where the symbols $n, M, d_L$ are the length, the number of the codewords and the minimum Lee distance of $\mathcal {C}$, respectively. Then, by Definition 1 and Proposition 1, we have that $\Phi(\mathcal {C})$ is a $(2n, M)$ linear code with minimum Hamming distance $d_L$ over $\mathbb{Z}_9$.

Let $\textbf{x}=(x_1, x_2, \ldots, x_n)$ and $\textbf{y}=(y_1, y_2, \ldots, y_n)$ be two vectors of $R^n$. The Euclidean inner product of $\textbf{x}$ and $\textbf{y}$ is defined as
\begin{equation*}
\textbf{x}\cdot \textbf{y}=\textbf{x}\textbf{y}^t=\sum_{i=1}^nx_iy_i.
\end{equation*}
The Euclidean dual code $\mathcal {C}^\perp$ of $\mathcal {C}$ is defined as $\mathcal {C}^\perp=\{\textbf{x}\in R^n| \textbf{x}\cdot \textbf{c}=0~{\rm for ~all~}\textbf{c}\in \mathcal {C}\}$. $\mathcal {C}$ is said to be Euclidean self-orthogonal if $\mathcal {C}\subseteq \mathcal {C}^\perp$ and Euclidean self-dual if $\mathcal {C}=\mathcal {C}^\perp$.
\begin{proposition}
Let $\mathcal {C}$ be a linear code of length $n$. Let $A$ be an element of ${\rm GL}_2(\mathbb{Z}_9)$ with $AA^t=\lambda I$, where $\lambda$ is a unit of $\mathbb{Z}_9$ and $I$ is the identity matrix of order $2$ over $\mathbb{Z}_9$. Then $\Phi(\mathcal {C})^\perp=\Phi(\mathcal {C}^\perp)$. Moreover, if $\mathcal {C}$ is Euclidean self-dual, so is $\Phi(\mathcal {C})$.
\end{proposition}
\begin{proof}
For all $\textbf{c}_1=(c_{1,0}, c_{1,1}, \ldots, c_{1,n-1})\in \mathcal {C}$ and $\textbf{c}_2=(c_{2,0}, c_{2,1}, \ldots, c_{2,n-1})\in \mathcal {C}^\perp$, where $c_{j,i}=ua_{j,i}+(1-u)b_{j,i}$, $a_{j,i}, b_{j,i}\in \mathbb{Z}_9$, $j=1,2$, $i=0,1,\ldots,n-1$, if $\textbf{c}_1\cdot \textbf{c}_2=0$, then we have $\textbf{c}_1\cdot \textbf{c}_2=\textbf{c}_1\textbf{c}_2^t=\sum_{i=0}^{n-1}c_{1,i}c_{2,i}=u\sum_{i=0}^{n-1}a_{1,i}a_{2,i}+(1-u)\sum_{i=0}^{n-1}b_{1,i}b_{2,i}=0$. Therefore, $\Phi(\textbf{c}_1)\cdot \Phi(\textbf{c}_2)=\Phi(\textbf{c}_1)\Phi(\textbf{c}_2)^t=(c_{1,0}A, c_{1,1}A, \ldots, c_{1,n-1}A)(c_{2,0}A, c_{2,1}A, \ldots, c_{2,n-1}A)^t=\sum_{i=0}^{n-1}(c_{1,i}A)(c_{2,i}A)^t=\lambda \sum_{i=0}^{n-1}c_{1,i}c_{2,i}^t=\lambda(u\sum_{i=0}^{n-1}a_{1,i}a_{2,i}+(1-u)\sum_{i=0}^{n-1}b_{1,i}b_{2,i})=0$. Thus $\Phi(\mathcal {C}^\perp)\subseteq \Phi(\mathcal {C})^\perp$. Since $|\Phi(\mathcal {C}^\perp)|=|\Phi(\mathcal {C})^\perp|$, it follows that $\Phi(\mathcal {C})^\perp=\Phi(\mathcal {C}^\perp)$. Clearly, $\Phi(\mathcal {C})$ is Euclidean self-orthogonal if $\mathcal {C}$ is Euclidean self-dual. Further, $|\Phi(\mathcal {C})|=|\mathcal {C}|=81^{n/2}=9^{2n/2}$. Thus, $\Phi(\mathcal {C})$ is Euclidean self-dual.
\end{proof}

Define $$\mathcal {C}_1=\{\textbf{x}\in \mathbb{Z}_9^n| \exists \textbf{y}\in \mathbb{Z}_9^n, u\textbf{x}+(1-u)\textbf{y}\in \mathcal {C}\}$$ and $$\mathcal {C}_2=\{\textbf{y}\in \mathbb{Z}_9^n| \exists \textbf{x}\in \mathbb{Z}_9^n, u\textbf{x}+(1-u)\textbf{y}\in \mathcal {C}\}.$$ Then $\mathcal {C}_1$ and $\mathcal {C}_2$ are both $\mathbb{Z}_9$-linear of length $n$. Moreover, the linear code $\mathcal {C}$ of length $n$ over $R$ can be expressed uniquely as $$\mathcal {C}=u\mathcal {C}_1\oplus (1-u)\mathcal {C}_2.$$
\begin{proposition}
Let $\mathcal {C}$ be a linear code of length $n$ over $R$. Then $\mathcal {C}^\perp=u\mathcal {C}_1^\perp\oplus(1-u)\mathcal {C}_2^\perp$. Moreover, $\mathcal {C}$ is Euclidean self-dual if and only if $\mathcal {C}_1$ and $\mathcal {C}_2$ are both Euclidean self-dual over $\mathbb{Z}_9$.
\end{proposition}
\begin{proof}
Define
\begin{equation*}
\widehat{\mathcal {C}}_1=\{ \textbf{x} \in \mathbb{Z}_9^n|~\exists \textbf{y} \in \mathbb{Z}_9^n, u\textbf{x}+(1-u)\textbf{y}\in \mathcal {C}^\perp\}
\end{equation*}
and
\begin{equation*}
\widehat{\mathcal {C}}_2=\{ \textbf{y} \in \mathbb{Z}_9^n|~\exists \textbf{x} \in \mathbb{Z}_9^n, u\textbf{x}+(1-u)\textbf{y}\in \mathcal {C}^\perp\}.
\end{equation*}
Then $\mathcal {C}^\perp=u\widehat{\mathcal {C}}_1+(1-u)\widehat{\mathcal {C}}_2$ and this expression is unique. Clearly, $\widehat{\mathcal {C}}_1\subseteq \mathcal {C}_1^\perp$. Let $\textbf{c}_1$ be an element of $\mathcal {C}_1^\perp$. Then, for any $\textbf{x}\in \mathcal {C}_1$, there exists $\textbf{y}\in \mathbb{Z}_9^n$ such that $\textbf{c}_1\cdot(u\textbf{x}+(1-u)\textbf{y})=\textbf{0}$. Let $\textbf{c}=u\textbf{x}+(1-u)\textbf{y}\in \mathcal {C}$. Then $u\textbf{c}_1\cdot \textbf{c}=\textbf{0}$, which implies that $u\textbf{c}_1\in \mathcal {C}^\perp$. By the unique expression of $\mathcal {C}^\perp$, we have $\textbf{c}_1\in \widehat{\mathcal {C}}_1$, i.e. $\mathcal {C}_1=\widehat{\mathcal {C}}_1$. Similarly, we can prove $\mathcal {C}_2=\widehat{\mathcal {C}}_2$ implying $\mathcal {C}^\perp=u\mathcal {C}_1^\perp+(1-u)\mathcal {C}_2^\perp$.
\par
Clearly, $\mathcal {C}$ is Euclidean self-dual over $R$ if $\mathcal {C}_1$ and $\mathcal {C}_2$ are both Euclidean self-dual over $\mathbb{Z}_9$. If $\mathcal {C}$ is Euclidean self-dual, then $\mathcal {C}_1$ and $\mathcal {C}_2$ are both Euclidean self-orthogonal over $\mathbb{Z}_9$, i.e. $\mathcal {C}_1\subseteq \mathcal {C}_1^\perp$ and $\mathcal {C}_2\subseteq \mathcal {C}_2^\perp$. Next, we will prove $\mathcal {C}_1=\mathcal {C}_1^\perp$ and  $\mathcal {C}_2=\mathcal {C}_2^\perp$. If not, then there are elements $\textbf{a}\in \mathcal {C}_1^\perp \setminus \mathcal {C}_1$ and $\textbf{b}\in \mathcal {C}_2$ such that $(u\textbf{a}+(1-u)\textbf{b})^2\neq \textbf{0}$, which is a contradiction that $\mathcal {C}$ is Euclidean self-dual. Therefore, $\mathcal {C}_1=\mathcal {C}_1^\perp$ and $\mathcal {C}_2=\mathcal {C}_2^\perp$.
\end{proof}
\section{Quadratic residue codes over $R$}
Let $T$ be the cyclic shift operator on $R^n$, i.e. for any vector $\textbf{c}=(c_0, c_1, \ldots, c_{n-1})$ of $ R^n$, $T(\textbf{c})=(c_{n-1}, c_0, \ldots, c_{n-2})$.
\par
A linear code $\mathcal {C}$ of length $n$ over $R$ is called cyclic if and only if $T(\mathcal {C})=\mathcal {C}$.
Define the polynomial ring $R_n=R[X]/(X^n-1)=\{c_0+c_1X+\cdots+c_{n-1}X^{n-1}+(X^n-1)|~c_0, c_1, \ldots, c_{n-1}\in R\}$. For any polynomial $c(X)+(X^n-1)\in R_n$, we denote it as $c(X)$ for simplicity.
\par
Define a map as follows
\begin{equation*}
\begin{split}
\varphi:~~R^n&\rightarrow R_n=R[X]/(X^n-1)\\
 (c_0, c_1, \ldots, c_{n-1})&\mapsto c(X)=c_0+c_1X+\cdots +c_{n-1}X^{n-1}.
 \end{split}
 \end{equation*}
Clearly, $\varphi$ is an $R$-module isomorphism from $R^n$ to $R_n$. And a linear code $\mathcal {C}$ of length $n$ is cyclic over $R$ if and only if $\varphi(\mathcal {C})$ is an ideal of $R_n$. Sometimes, we identify the cyclic code $\mathcal {C}$ with the ideal of $R_n$.
\begin{lemma}
A linear code $\mathcal {C}=u\mathcal {C}_1\oplus (1-u)\mathcal {C}_2$ is cyclic over $R$ if and only if $\mathcal {C}_1$ and $\mathcal {C}_2$ are both cyclic over $\mathbb{Z}_9$.
\end{lemma}
\begin{proof}
Let $(a_0, a_1, \ldots, a_{n-1}) \in \mathcal {C}_1$ and $(b_0, b_1, \ldots, b_{n-1})\in \mathcal {C}_2$. Assume that $c_i=ua_i+(1-u)b_i$ for $i=0,1,\ldots, n-1$. Then the vector $(c_0, c_1, \ldots, c_{n-1})$ belongs to $\mathcal {C}$. Since $\mathcal {C}$ is a cyclic code, it follows that $(c_{n-1}, c_0, \ldots, c_{n-2})\in \mathcal {C}$. Note that $(c_{n-1}, c_0, \ldots, c_{n-2})=u(a_{n-1}, a_0, \ldots, a_{n-2})+(1-u)(b_{n-1}, b_0, \ldots, b_{n-2})$. Hence $(a_{n-1}, a_0, \ldots, a_{n-2})\in \mathcal {C}_1$ and $(b_{n-1}, b_0, \ldots, b_{n-2})\in \mathcal {C}_2$, which implies that $\mathcal {C}_1$ and $\mathcal {C}_2$ are both cyclic codes over $\mathbb{Z}_9$.
\par
Conversely, let $\mathcal {C}_1$ and $\mathcal {C}_2$ be both cyclic codes over $\mathbb{Z}_9$. Let $(c_0, c_1, \ldots, c_{n-1})\in \mathcal {C}$, where $c_i=ua_i+(1-u)b_i$ for $i=0,1,\ldots,n-1$. Then $(a_0, a_1, \ldots, a_{n-1})\in\mathcal {C}_1$ and $(b_0, b_1, \ldots, b_{n-1})\in\mathcal {C}_2$. Note that $(c_{n-1}, c_0, \ldots, c_{n-2})=u(a_{n-1}, a_0, \ldots, a_{n-2})+(1-u)(b_{n-1}, b_0, \ldots, b_{n-2})\in u\mathcal {C}_1\oplus(1-u)\mathcal {C}_2=\mathcal {C}$. Therefore, $\mathcal {C}$ is a cyclic code over $R$.
\end{proof}

In the following of this section, we assume that $n$ is a positive integer such that ${\rm gcd}(n, 9)=1$. Let $C$ be a cyclic code of length $n$ over $\mathbb{Z}_9$. Then, similar to the cyclic codes over $\mathbb{Z}_4$ in \cite{Wan}, there exist unique monic polynomials $f(X), g(X), h(X)$ such that $X^n-1=f(X)g(X)h(X)$ and $C=(f(X)g(X),3f(X)h(X))$.
\begin{lemma}
Let $\mathcal {C}=u\mathcal {C}_1\oplus (1-u)\mathcal {C}_2$ be a cyclic code of length $n$ over $R$. Then $\mathcal {C}=(uf_1(X)g_1(X)+(1-u)f_2(X)g_2(X), 3uf_1(X)h_1(X)+3(1-u)f_2(X)h_2(X))$, where $f_1(X)g_1(X)h_1(X)=f_2(X)g_2(X)h_2(X)=X^n-1$ and $\mathcal {C}_1=(f_1(X)g_1(X),3f_1(X)h_1(X))$, $\mathcal {C}_2=(f_2(X)g_2(X),3f_2(X)h_2(X))$ over $\mathbb{Z}_9$, respectively.
\end{lemma}
\begin{proof}
Let $\widetilde{\mathcal {C}}=(uf_1(X)g_1(X)+(1-u)f_2(X)g_2(X), 3uf_1(X)h_1(X)+3(1-u)f_2(X)h_2(X))$, $\mathcal {C}_1=(f_1(X)g_1(X),3f_1(X)h_1(X))$ and $\mathcal {C}_2=(f_2(X)g_2(X),3f_2(X)h_2(X))$. Clearly, $\widetilde{\mathcal {C}}\subseteq \mathcal {C}$. For $u\mathcal {C}_1$, we have $u\mathcal {C}_1=u\widetilde{\mathcal {C}}$ since $u^2=u$ over $\mathbb{Z}_9$. Similarly, $(1-u)\mathcal {C}_2=(1-u)\widetilde{\mathcal {C}}$. Therefore $u\mathcal {C}_1\oplus (1-u)\mathcal {C}_2\subseteq \widetilde{\mathcal {C}}$. Thus $\mathcal {C}=\widetilde{\mathcal {C}}$.
\end{proof}
\begin{proposition}
The quotient polynomial ring $R[X]/(X^n-1)$ is principal.
\end{proposition}
\begin{proof}
Let $C=(f(X)g(X), 3f(X)h(X))$ be a cyclic code of length $n$ over $\mathbb{Z}_9$, where $X^n-1=f(X)g(X)h(X)$. Then, similar to the cyclic codes over $\mathbb{Z}_4$ in \cite{Wan}, $C=(f(X)g(X)+3f(X))$. By Lemma 2, we have  any cyclic code $\mathcal {C}$ is principal over $R$, which implies the result.
\end{proof}

Furthermore, the number of distinct cyclic codes of length $n$ over $R$ is $9^r$, where $r$ is the number of the basic irreducible factors of $X^n-1$ over $\mathbb{Z}_9$.

\begin{lemma}
Let $\mathcal {C}=(uf_1(X)g_1(X)+(1-u)f_2(X)g_2(X), 3uf_1(X)h_1(X)+3(1-u)f_2(X)h_2(X))$, where $f_1(X)g_1(X)h_1(X)=f_2(X)g_2(X)h_2(X)=X^n-1$ and $\mathcal {C}_1=(f_1(X)g_1(X), 3f_1(X)h_1(X))$, $\mathcal {C}_2=(f_2(X)g_2(X),3f_2(X)h_2(X))$ over $\mathbb{Z}_9$, respectively. Then $\mathcal {C}$ is Euclidean self-dual if and only if $f_1(X)=h_1^*(X), g_1(X)=g_1^*(X)$ and $f_2(X)=h_2^*(X), g_2(X)=g_2^*(X)$, where $f^*(X)=\pm X^{{\rm deg}f(X)}f(X^{-1})$.
\end{lemma}
\begin{proof}
First, by $\mathcal {C}^\perp=u\mathcal {C}_1^\perp\oplus (1-u)\mathcal {C}_2^\perp$, we have $\mathcal {C}^\perp$ is also a cyclic code if $\mathcal {C}$ is a cyclic code. Moreover, by Proposition 3, we have $\mathcal {C}$ is Euclidean self-dual over $R$ if and only if $\mathcal {C}_1$ and $\mathcal {C}_2$ are both Euclidean self-dual over $\mathbb{Z}_9$. Then, similar to Theorem 12.5.10 in \cite{Huffuman}, we deduce the result.
\end{proof}
By Proposition 3 and Theorem 4.4 in \cite{Batoul}, we have the following result immediately.
\begin{lemma}
Non-zero Euclidean cyclic self-dual codes of length $n$ exist over $R$ if and only if $3^j\not\equiv -1~({\rm mod}n)$ for any $j$.
\end{lemma}
\par
In the following, we consider some properties of the generating idempotents of cyclic codes over $R$. An element $e(X)\in \mathcal {C}$ is called an idempotent element if $e(X)^2=e(X)$ in $R_n$.
\begin{lemma}
Let $\mathcal {C}$ be a cyclic code of length $n$. Then there exists a unique idempotent element $e(X)=ue_1(X)+(1-u)e_2(X)\in R[X]$ such that $\mathcal {C}=(e(X))$.
\end{lemma}
\begin{proof}
Since ${\rm gcd}(n,9)=1$, it follows that there exist unique idempotent elements $e_1(X), e_2(X) \in \mathbb{Z}_9[X]$ such that $\mathcal {C}_1=(e_1(X))$ and $\mathcal {C}_2=(e_2(X))$. By Lemma 2, we have $\mathcal {C}=(ue_1(X)+(1-u)e_2(X))$. Let $e(X)=ue_1(X)+(1-u)e_2(X)$. Then $e(X)^2=ue_1(X)^2+(1-u)e_2(X)^2=ue_1(X)+(1-u)e_2(X)=e(X)$, which implies that $e(X)$ is an idempotent element of $\mathcal {C}$. If there is another $d(X)\in \mathcal {C}$ such that $\mathcal {C}=(d(X))$ and $d(X)^2=d(X)$. Since $d(X)\in \mathcal {C}=(e(X))$, we have that $d(X)=a(X)e(X)$ for some $a(X)\in R_n$. And then, $d(X)e(X)=a(X)e(X)^2=d(X)$. Similarly, we can prove $d(X)e(X)=e(X)$, which implies that $e(X)$ is unique.
\end{proof}

The idempotent element $e(X)$ in Lemma 5 is called the generating idempotent of $\mathcal {C}$.
\begin{lemma}
Let $\mathcal {C}=u\mathcal {C}_1\oplus (1-u)\mathcal {C}_2$ be a cyclic code of length $n$ over $R$. Let $e(X)=ue_1(X)+(1-u)e_2(X)$, where $e_1(X)$ and $e_2(X)$ are generating idempotents of $\mathcal {C}_1$ and $\mathcal {C}_2$ over $\mathbb{Z}_9$, respectively. Then the Euclidean dual code $\mathcal {C}^\perp$ has $1-e(X^{-1})$ as its generating idempotent.
\end{lemma}
\begin{proof}
By Proposition 3, we have $\mathcal {C}^\perp=u\mathcal {C}_1^\perp \oplus (1-u)\mathcal {C}_2^\perp$. Moreover, $\mathcal {C}^\perp$ is also a cyclic code since $\mathcal {C}_1^\perp$ and $\mathcal {C}_2^\perp$ are both cyclic codes. Let $e_1(X)$ and $e_2(X)$ be generating idempotents of $\mathcal {C}_1$ and $\mathcal {C}_2$, respectively. Then $\mathcal {C}_1^\perp$ and $\mathcal {C}_2^\perp$ have $1-e_1(X^{-1})$ and $1-e_2(X^{-1})$ as their generating idempotents respectively. Let $\widetilde{e}(X)$ be the generating idempotent of $\mathcal {C}^\perp$. Then, by Lemma 5, $\widetilde{e}(X)=u(1-e_1(X^{-1}))+(1-u)(1-e_2(X^{-1}))=1-e(X^{-1})$.
\end{proof}

Let $p$ be a prime number with $p\equiv \pm 1 ({\rm mod}12)$. Let $\mathcal {Q}_p$ denote the set of nonzero quadratic residues modulo $p$, and let $\mathcal {N}_p$ be the set of quadratic non-residues modulo $p$.
\par
Let $Q(X)=\sum_{i\in \mathcal {Q}_p}X^i$, $N(X)=\sum_{i\in \mathcal {N}_p}X^i$ and $J(X)=1+Q(X)+N(X)$. By Lemma 5 and Theorem 6 in \cite{Taeri}, we have the following results immediately.
\begin{lemma}
 \vskip 3mm \noindent
   {\bf  I:~~} Suppose that $p=12r-1$ and $k$ is a positive integer.\\
 {\rm (i)}~If $r=3k$, let the set $S_0$ be $\{8Q(X), 8N(X), 8J(X), 1+Q(X), 1+N(X), 1+2J(X)\}$.\\
 {\rm (ii)}~If $r=3k+1$, let the set $S_1$ be $\{3+6Q(X)+8N(X), 3+6N(X)+8Q(X), 7+Q(X)+3N(X), 7+N(X)+3Q(X), 5J(X), 1+5J(X)\}$.\\
 {\rm (iii)}~If $r=3k+2$, let the set $S_2$ be $\{6+3Q(X)+8N(X), 6+3N(X)+8Q(X), 4+Q(X)+6N(X), 4+N(X)+6Q(X), 2J(X), 8+8J(X)\}$.
  \vskip 3mm \noindent
   {\bf  II:~~}Suppose that $p=12r+1$ and $k$ is a positive integer\\
 {\rm (i)}~If $r=3k$, let the set $T_0$ be $\{1+Q(X), 1+N(X), 8Q(X), 8N(X), J(X), 1+J(X)\}$.\\
 {\rm (ii)}~If $r=3k+1$, let the set $T_1$ be $\{4+Q(X)+6N(X), 4+N(X)+6Q(X), 6+3Q(X)+8N(X), 6+3N(X)+8Q(X), 7J(X), 1+2J(X)\}$.\\
 {\rm (iii)}~If $r=3k+2$, let the set $T_2$ be $\{7+Q(X)+3N(X), 7+N(X)+3Q(X), 3+6Q(X)+8N(X), 3+6N(X)+8Q(X), 4J(X), 1+5J(X)\}$.\\
 Then, for any $e_1(X), e_2(X)\in S_0$ or $S_1$ or $S_2$ or $T_0$ or $T_1$ or $T_2$, $e(X)=ue_1(X)+(1-u)e_2(X)$ is the idempotent of $R_p$.
\end{lemma}

We now discuss the quadratic residue codes over $R$. First, we give the definitions of these codes. The definitions depend upon the value $p$ modulo $12$.
\begin{definition}
Let $p=12r-1$. \\
{\rm (i)}~If $r=3k$, define
\begin{equation*}
\mathcal {D}_1=(u(8Q(X))+(1-u)(8N(X))),
\end{equation*}
\begin{equation*}
\mathcal {D}_2=(u(8N(X))+(1-u)(8Q(X))),
\end{equation*}
\begin{equation*}
\mathcal {E}_1=(u(1+N(X))+(1-u)(1+Q(X))),
\end{equation*}
\begin{equation*}
\mathcal {E}_2=(u(1+Q(X))+(1-u)(1+N(X))).
\end{equation*}
{\rm (ii)}~If $r=3k+1$, define
\begin{equation*}
\mathcal {D}_1=(u(3+6Q(X)+8N(X))+(1-u)(3+6N(X)+8Q(X))),
\end{equation*}
\begin{equation*}
\mathcal {D}_2=(u(3+6N(X)+8Q(X))+(1-u)(3+6Q(X)+8N(X))),
\end{equation*}
\begin{equation*}
\mathcal {E}_1=(u(7+Q(X)+3N(X))+(1-u)(7+N(X)+3Q(X))),
\end{equation*}
\begin{equation*}
\mathcal {E}_2=(u(7+N(X)+3Q(X))+(1-u)(7+Q(X)+3N(X))).
\end{equation*}
{\rm (iii)}~If $r=3k+2$, define
\begin{equation*}
\mathcal {D}_1=(u(6+3Q(X)+8N(X))+(1-u)(6+3N(X)+8Q(X))),
\end{equation*}
\begin{equation*}
\mathcal {D}_2=(u(6+3N(X)+8Q(X))+(1-u)(6+3Q(X)+8N(X))),
\end{equation*}
\begin{equation*}
\mathcal {E}_1=(u(4+Q(X)+6N(X))+(1-u)(4+N(X)+6Q(X))),
\end{equation*}
\begin{equation*}
\mathcal {E}_2=(u(4+N(X)+6Q(X))+(1-u)(4+Q(X)+6N(X))).
\end{equation*}
These cyclic codes of length $p$ are called the quadratic residue codes over $R$ at the case I.
\end{definition}
\par\vskip 2mm
Let $a$ be a non-zero positive integer defined as $\mu_a(i)=ai$ for any positive integer $i$. This map acts on polynomials as
\begin{equation*}
\mu_a(\sum_iX^i)=\sum_iX^{ai}.
\end{equation*}
\begin{theorem}
Let $p=12r-1$. Then the quadratic residue codes defined above satisfy the following:\\
{\rm (i)}~$\mathcal {D}_i\mu_a=\mathcal {D}_i$ and $\mathcal {E}_i\mu_a=\mathcal {E}_i$ for $i=1,2$ and $a\in \mathcal {Q}_p$; $\mathcal {D}_1\mu_a=\mathcal {D}_2$ and $\mathcal {E}_1\mu_a=\mathcal {E}_2$ for $a\in \mathcal {N}_p$.\\
{\rm (ii)}~$\mathcal {D}_1\cap \mathcal {D}_2=(K(X))$ and $\mathcal {D}_1+\mathcal {D}_2=R_p$, where $K(X)$ is a suitable element of $\{2J(X), 5J(X), 8J(X)\}$.\\
{\rm (iii)}~$\mathcal {E}_1\cap\mathcal {E}_2=\{0\}$ and $\mathcal {E}_1 + \mathcal {E}_2=(K(X))^\perp$.\\
{\rm (iv)}~$|\mathcal {D}_1|=|\mathcal {D}_2|=9^{p+1}$ and $|\mathcal {E}_1|=|\mathcal {E}_2|=9^{p-1}$.\\
{\rm (v)}~$\mathcal {D}_i=\mathcal {E}_i+(K(X))$ for $i=1,2$.\\
{\rm (vi)}~$\mathcal {E}_1$ and $\mathcal {E}_2$ are Euclidean self-orthogonal and $\mathcal {E}_i^\perp=\mathcal {D}_i$ for $i=1,2$.
\end{theorem}
\begin{proof}
Let $p=12r-1$. We only verify when $r=3k$. The proof of other cases are similar.

(i)~If $a\in \mathcal {Q}_p$, then $(u(8Q(X))+(1-u)(8N(X))\mu_a=u(8Q(X))+(1-u)(8N(X))$, which implies that $\mathcal {D}_1\mu_a=\mathcal {D}_1$. Similarly, $\mathcal {D}_2\mu_a=\mathcal {D}_2$.

If $a\in \mathcal {N}_p$, then $(u(8Q(X))+(1-u)(8N(X))\mu_a=u(8N(X))+(1-u)(8Q(X))$, which implies that $\mathcal {D}_1\mu_a=\mathcal {D}_2$.

The parts of (i) involving $\mathcal {E}_i$ are similar.

(ii)~Since $u(8Q(X))+(1-u)(8N(X))+u(8N(X))+(1-u)(8Q(X))=8J(X)+1$, it follows that $(u(8Q(X))+(1-u)(8N(X)))(8J(X))=(u(8Q(X))+(1-u)(8N(X)))(u(8N(X))+(1-u)(8Q(X)))$.

On the other hand, $(u(8Q(X))+(1-u)(8N(X)))(8J(X))=\frac{p-1}{2}J(X)$. Since $p=12r-1$ and $r=3k$, we have $\frac{p-1}{2}\equiv 8({\rm mod}9)$ implying
$(u(8Q(X))+(1-u)(8N(X)))(u(8N(X))+(1-u)(8Q(X)))=8J(X)$. It means that $\mathcal {D}_1\cap \mathcal {D}_2=(K(X))=(8J(X))$ and $\mathcal {D}_1+\mathcal {D}_2=R_p$.

(iii)~Since $u(1+N(X))+(1-u)(1+Q(X))+u(1+Q(X))+(1-u)(1+N(X))=1+J(X)$, it follows that $(u(1+N(X))+(1-u)(1+Q(X)))J(X)=(u(1+N(X))+(1-u)(1+Q(X)))(u(1+Q(X))+(1-u)(1+N(X)))$.

On the other hand, $(u(1+N(X))+(1-u)(1+Q(X)))J(X)=\frac{p+1}{2}J(X)$. Since $p=12r-1$ and $r=3k$, we have $\frac{p+1}{2}\equiv 0({\rm mod}9)$ implying $(u(1+N(X))+(1-u)(1+Q(X)))(u(1+Q(X))+(1-u)(1+N(X)))=0$. It means that $\mathcal {E}_1\cap \mathcal {E}_2=\{0\}$ and $\mathcal {E}_1+\mathcal {E}_2=(1+J(X))=(K(X))^\perp$.

(iv)~This follows from (i)(ii)(iii) immediately.

(v)~From (ii), we have $K(X)\in \mathcal {D}_2$ implying that $(u(8N(X))+(1-u)(8Q(X)))K(X)=K(X)$ as $u(8N(X))+(1-u)(8Q(X))$ is the multiplicative identity of $\mathcal {D}_2$. Then $\mathcal {E}_1+(K(X))=u(1+N(X))+(1-u)(1+Q(X))+K(X)-(u(1+N(X))+(1-u)(1+Q(X)))K(X)=u(1+N(X))+(1-u)(1+Q(X))+K(X)-(K(X)-(u(8N(X))+(1-u)(8Q(X)))K(X))=u(1+Q(X))+(1-u)(1+N(X))+K(X)
-(K(X)-K(X))=u(8Q(X))+(1-u)(8N(X))$, which implies that $\mathcal {E}_1+(K(X))=\mathcal {D}_1$. Similarly, $\mathcal {E}_2+(K(X))=\mathcal {D}_2$.

(vi)~From Lemma 14, the generating idempotent for $\mathcal {E}_1^\perp$ is $1-(u(1+N(X))+(1-u)(1+Q(X)))_{\mu_{-1}}=(u(8N(X))+(1-u)(8Q(X)))_{\mu_{-1}}$. Since $-1\in \mathcal {N}_p$ as $p=12r-1$, it follows that $N(X)_{\mu_{-1}}=Q(X)$ and $Q(X)_{\mu_{-1}}=N(X)$. Therefore the generating idempotent of $\mathcal {E}_1^\perp$ is $u(8Q(X))+(1-u)(8N(X))$, which implies that $\mathcal {E}_1^\perp=\mathcal {D}_1$. Similarly, $\mathcal {E}_2^\perp=\mathcal {D}_2$. From (v), $\mathcal {E}_i\subseteq \mathcal {D}_i$ implying $\mathcal {E}_i$ is self-orthogonal for $i=1,2$.
\end{proof}
\begin{definition}
Let $p=12r+1$.\\
(i)~If $r=3k$, define
\begin{equation*}
\mathcal {D}_1=(u(1+Q(X))+(1-u)(1+N(X)))
\end{equation*}
\begin{equation*}
\mathcal {D}_2=(u(1+N(X))+(1-u)(1+Q(X)))
\end{equation*}
\begin{equation*}
\mathcal {E}_1=(u(8N(X))+(1-u)(8Q(X)))
\end{equation*}
\begin{equation*}
\mathcal {E}_2=(u(8Q(X))+(1-u)(8N(X))).
\end{equation*}
(ii)~If $r=3k+1$, define
\begin{equation*}
\mathcal {D}_1=(u(4+Q(X)+6N(X))+(1-u)(4+6Q(X)+N(X))),
\end{equation*}
\begin{equation*}
\mathcal {D}_2=(u(4+6Q(X)+N(X))+(1-u)(4+Q(X)+6N(X))),
\end{equation*}
\begin{equation*}
\mathcal {E}_1=(u(6+3Q(X)+8N(X))+(1-u)(6+8Q(X)+3N(X))),
\end{equation*}
\begin{equation*}
\mathcal {E}_2=(u(6+8Q(X)+3N(X))+(1-u)(6+3Q(X)+8N(X))).
\end{equation*}
(iii)~If $r=3k+2$, define
\begin{equation*}
\mathcal {D}_1=(u(7+Q(X)+3N(X))+(1-u)(7+3Q(X)+N(X))),
\end{equation*}
\begin{equation*}
\mathcal {D}_2=(u(7+3Q(X)+N(X))+(1-u)(7+Q(X)+3N(X))),
\end{equation*}
\begin{equation*}
\mathcal {E}_1=(u(3+6Q(X)+8N(X))+(1-u)(3+8Q(X)+6N(X))),
\end{equation*}
\begin{equation*}
\mathcal {E}_2=(u(3+8Q(X)+6N(X))+(1-u)(3+6Q(X)+8N(X))).
\end{equation*}
These cyclic codes of length $p$ are called the quadratic residue codes over $R$ at the case II.
\end{definition}

Similar to Theorem 5, we also have the following results. Here we omit the proof.
\begin{theorem}
Let $p=12r+1$. Then the quadratic residue codes defined above satisfy the following:\\
{\rm (i)}~$\mathcal {D}_i\mu_a=\mathcal {D}_i$ and $\mathcal {E}_i\mu_a=\mathcal {E}_i$ for $i=1,2$ and $a\in \mathcal {Q}_p$; $\mathcal {D}_1\mu_a=\mathcal {D}_2$ and $\mathcal {E}_1\mu_a=\mathcal {E}_2$ for $a\in \mathcal {N}_p$.\\
{\rm (ii)}~$\mathcal {D}_1\cap \mathcal {D}_2=(K(X))$ and $\mathcal {D}_1+\mathcal {D}_2=R_p$, where $K(X)$ is a suitable element of $\{J(X), 4J(X), 7J(X)\}$.\\
{\rm (iii)}~$\mathcal {E}_1\cap\mathcal {E}_2=\{0\}$ and $\mathcal {E}_1 + \mathcal {E}_2=(K(X))^\perp$.\\
{\rm (iv)}~$|\mathcal {D}_1|=|\mathcal {D}_2|=9^{p+1}$ and $|\mathcal {E}_1|=|\mathcal {E}_2|=9^{p-1}$.\\
{\rm (v)}~$\mathcal {D}_i=\mathcal {E}_i+(K(X))$ for $i=1,2$.\\
{\rm (vi)}~$\mathcal {E}_1^\perp=\mathcal {D}_2$ and $\mathcal {E}_2^\perp=\mathcal {D}_1$.
\end{theorem}

Let $\mathcal {D}_1$ and $\mathcal {D}_2$ be the quadratic residue codes defined above. In the following, we discuss the extension of $\mathcal {D}_i$ denoted as $\widehat{\mathcal {D}}_i$ for $i=1,2$. Define $\widehat{\mathcal {D}}_i=\{(c_\infty, c_0,c_1,\ldots, c_{p-1})|(c_0,c_1,\ldots,c_{p-1})\in \mathcal {D}_i, c_\infty+c_0+\cdots+c_{p-1}\equiv 0 ({\rm mod}9)\}$. For $i=1,2$, $\widehat{\mathcal {D}}_i$ is the\emph{ extended quadratic residue code} of $\mathcal {D}_i$ with length $p+1$ over $R$.

From the fact that, for $i=1,2$, the sum of the components of any codeword in $\mathcal {E}_i$ is $0$, we have the following lemma immediately.
\begin{lemma}
For $i=1,2$, let $G_i$ and $\widehat{G}_i$ be the generator matrices of the quadratic residue codes $\mathcal {E}_i$ and $\widehat{\mathcal {D}}_i$, respectively.\\
{\rm (i)}~If $p=12r-1$, then
$$\widehat{G}_i=\left(\begin{array}{cccc}8&8&\cdots&8\\ 0& & &\\\vdots& & \huge{G}_i &\\0& & & \end{array}\right),~~\widehat{G}_i=\left(\begin{array}{cccc}8&5&\cdots&5\\ 0& & &\\\vdots& & \huge{G}_i &\\0& & & \end{array}\right)~~and~~\widehat{G}_i=\left(\begin{array}{cccc}8&2&\cdots&2\\ 0& & &\\\vdots& & \huge{G}_i &\\0& & & \end{array}\right)$$
corresponding to $r=3k$, $r=3k+1$ and $r=3k+2$ respectively.\\
{\rm(ii)}~If $p=12r+1$, then
$$\widehat{G}_i=\left(\begin{array}{cccc}8&1&\cdots&1\\ 0& & &\\\vdots& & \huge{G}_i &\\0& & & \end{array}\right),~~\widehat{G}_i=\left(\begin{array}{cccc}8&7&\cdots&7\\ 0& & &\\\vdots& & \huge{G}_i &\\0& & & \end{array}\right)~~and~~\widehat{G}_i=\left(\begin{array}{cccc}8&4&\cdots&4\\ 0& & &\\\vdots& & \huge{G}_i &\\0& & & \end{array}\right)$$
corresponding to $r=3k$, $r=3k+1$ and $r=3k+2$ respectively.
\end{lemma}

When $p=12r+1$, for $i=1,2$, we define $\widetilde{D}_i$ to be the linear code of length $p+1$ over $R$ with the generator matrix
$$\left(\begin{array}{cccc}1&1&\cdots&1\\ 0& & &\\\vdots& & \huge{G}_i &\\0& & & \end{array}\right).$$
\begin{theorem}
Let $\mathcal {D}_i$ be the quadratic residue codes of length $p$ over $R$. The following hold\\
{\rm (i)}~If $p=12r-1$, then $\widehat{\mathcal {D}}_i$ are Euclidean self-dual for $i=1,2$. \\
{\rm (ii)}~If $p=12r+1$, then $\widehat{\mathcal {D}}_1^\perp=\widetilde{\mathcal {D}}_2$ and $\widehat{\mathcal {D}}_2^\perp=\widetilde{\mathcal {D}}_1$.
\end{theorem}
\begin{proof}
If $p=12r-1$, by the fact that the sum of the components of any codeword in $\mathcal {E}_i$ is zero, we have $\widehat{\mathcal {D}}_i$ are Euclidean self-orthogonal for $i=1,2$. Furthermore, $|\mathcal {D}_i|=|\widehat{\mathcal {D}}_i|=9^{p+1}$ implying $\widehat{\mathcal {D}}_i$ is Euclidean self-dual for $i=1,2$.
\vskip 2mm
If $p=12r+1$, then $\mathcal {E}_1^\perp=\mathcal {D}_2$ and $\mathcal {E}_2^\perp=\mathcal {D}_1$. Hence the extended codewords arising from $\mathcal {E}_i$ are orthogonal to all codewords in either $\widehat{\mathcal {D}}_j$ or $\widetilde{\mathcal {D}}_j$ where $j\neq i$. We prove the case $r=3k$. The proofs of other cases are similar. Since the product of the vectors $(8,1,\ldots, 1)$ and $(1,1,\ldots,1)$ is $8+p\equiv 0({\rm mod}9)$, we have $\widehat{\mathcal {D}}_j^\perp\subseteq \widetilde{\mathcal {D}}_i$ where $j\neq i$. Furthermore, $|\mathcal {D}_i|=|\widehat{\mathcal {D}}_i|=|\widetilde{\mathcal {D}}_i|=9^{p+1}$ implying $\widehat{\mathcal {D}}_j^\perp=\widetilde{\mathcal {D}}_i$ where $j\neq i$.
\end{proof}
\section{Examples}
In this section, we give some examples to illustrate the main work in this paper. Let $A=\left(\begin{array}{cc}1 & 1\\-1 & 1  \\\end{array}\right)$ be an matrix of ${\rm GL}_2(\mathbb{Z}_9)$. Then we have that $AA^t=2I$. Let $\mathcal {C}$ be a Euclidean self-dual code of length $n$ over $R$ and $\Phi$ be the Gray map corresponding to $A$. Then, by Proposition 2,  $\Phi(\mathcal {C})$ is a Euclidean self-dual code of length $2n$ over $\mathbb{Z}_9$.
\begin{example}
By Lemma 12, there exists a Euclidean cyclic self-dual code of length $11$ over $R$. It is well known that
\begin{equation*}
X^{11}-1=(X-1)(X^5+3X^4+8X^3+X^2+2X-1)(X^5-2X^4-X^3+X^2-3X-1).
\end{equation*}
Let $g(X)=1-X$ and $f(X)=X^5+3X^4+8X^3+X^2+2X-1$. Then $X^{11}-1=f(X)f^*(X)g(X)$. Assume that $\mathcal {C}_1=\mathcal {C}_2=(f^*(X)g(X), 3f(X)f^*(X))$, then by Lemma 10 we have that the cyclic code $\mathcal {C}=(f^*(X)g(X), 3f(X)f^*(X))$ is Euclidean self-dual of length $11$ over $R$. As a $\mathbb{Z}_9$-module, $\mathcal {C}$ is generated by the following matrix
\begin{equation*}
\left(
\begin{array}{ccccccccccc}
u & 2u & 5u & 2u & u & 6u & u & 0 & 0 & 0 & 0 \\
0 & u & 2u & 5u & 2u & u & 6u & u & 0 & 0 & 0 \\
0 & 0 & u & 2u & 5u & 2u & u & 6u & u & 0 & 0 \\
0 & 0 & 0 & u & 2u & 5u & 2u & u & 6u & u & 0 \\
0 & 0 & 0 & 0 & u & 2u & 5u & 2u & u & 6u & u \\
6u & 6u & 6u & 6u & 6u & 6u & 6u & 6u & 6u & 6u & 6u \\
1-u & 2-2u & 5-5u & 2-2u & 1-u & 6-6u & 1-u & 0 & 0 & 0 & 0 \\
0 &1-u & 2-2u & 5-5u & 2-2u & 1-u & 6-6u & 1-u & 0 & 0 & 0 \\
0 & 0 &1-u & 2-2u & 5-5u & 2-2u & 1-u & 6-6u & 1-u & 0 & 0 \\
0 & 0 & 0 &1-u & 2-2u & 5-5u & 2-2u & 1-u & 6-6u & 1-u & 0 \\
0 & 0 & 0 & 0 &1-u & 2-2u & 5-5u & 2-2u & 1-u & 6-6u & 1-u \\
6-6u & 6-6u & 6-6u  & 6-6u & 6-6u & 6-6u & 6-6u & 6-6u & 6-6u & 6-6u & 6-6u \\
\end{array}
\right).
\end{equation*}
By Proposition 2, we have that $\Phi(\mathcal {C})$ is a $\mathbb{Z}_9$-Euclidean self-dual code of length $22$ with minimum Hamming weight $5$. The Hamming weight distribution of $\Phi(\mathcal {C})$ is \\
$1+264x^5+1056x^6+7920x^7+32340x^8+81400x^9+152064x^{10}+236064x^{11}+1324224x^{12}
+8450640x^{13}+43501920x^{14}+188818080x^{15}+667663524x^{16}+1900455216x^{17}+
4216439920x^{18}+7043034240x^{19}+8466532800x^{20}+6507959040x^{21}+2336368896x^{22}$.
\end{example}

\begin{example}
Let $p=11$. We consider the quadratic residue codes of length $11$ over $R$. By the definitions of $Q(X)$ and $N(X)$, we have that $Q(X)=X+X^3+X^4+X^5+X^9$ and $N(X)=X^2+X^6+X^7+X^8+X^{10}$. Since $11=12\times 1-1$, by the Definition 2(ii), it follows that
\begin{equation*}
\begin{split}
\mathcal {D}_1&=(u(3+6X+8X^2+6X^3+6X^4+6X^5+8X^6+8X^7+8X^8+6X^9+8X^{10}) \\
              &+(1-u)(3+8X+6X^2+8X^3+8X^4+8X^5+6X^6+6X^7+6X^8+8X^9+6X^{10}))
 \end{split}
\end{equation*}
\begin{equation*}
\begin{split}
\mathcal {D}_2&=(u(3+8X+6X^2+8X^3+8X^4+8X^5+6X^6+6X^7+6X^8+8X^9+6X^{10}) \\
              &+(1-u)(3+6X+8X^2+6X^3+6X^4+6X^5+8X^6+8X^7+8X^8+6X^9+8X^{10}))
 \end{split}
\end{equation*}
\begin{equation*}
\begin{split}
\mathcal {E}_1&=(u(7+X+3X^2+X^3+X^4+X^5+3X^6+3X^7+3X^8+X^9+3X^{10}) \\
              &+(1-u)(7+3X+X^2+3X^3+3X^4+3X^5+X^6+X^7+X^8+3X^9+X^{10}))
 \end{split}
\end{equation*}
\begin{equation*}
\begin{split}
\mathcal {E}_2&=(u(7+3X+X^2+3X^3+3X^4+3X^5+X^6+X^7+X^8+3X^9+X^{10}) \\
              &+(1-u)(7+X+3X^2+X^3+X^4+X^5+3X^6+3X^7+3X^8+X^9+3X^{10}))
 \end{split}
\end{equation*}
are quadratic residue codes of length $11$ over $R$. From Lemma 13, the codes $\mathcal {E}_1$ and $\mathcal {E}_2$ can be regarded as the $\mathbb{Z}_9[X]$-modules, i.e.,
\begin{equation*}
\begin{split}
\mathcal {E}_1&=(u(7+X+3X^2+X^3+X^4+X^5+3X^6+3X^7+3X^8+X^9+3X^{10})) \\
              &\oplus ((1-u)(7+3X+X^2+3X^3+3X^4+3X^5+X^6+X^7+X^8+3X^9+X^{10}))
 \end{split}
\end{equation*}
\begin{equation*}
\begin{split}
\mathcal {E}_2&=(u(7+3X+X^2+3X^3+3X^4+3X^5+X^6+X^7+X^8+3X^9+X^{10})) \\
              &\oplus ((1-u)(7+X+3X^2+X^3+X^4+X^5+3X^6+3X^7+3X^8+X^9+3X^{10})),
 \end{split}
\end{equation*}
which implies that $\mathcal {E}_1$ and $\mathcal {E}_2$ have the following $\mathbb{Z}_9$-generator matrices respectively
$$G_1=\left(\begin{array}{c}uG_{11}\\ (1-u)G_{12} \end{array}\right)~~{\rm and}~~G_2=\left(\begin{array}{c}uG_{21}\\ (1-u)G_{22} \end{array}\right),$$
where
\begin{equation*}
G_{11}=G_{22}=\left(
\begin{array}{ccccccccccc}
1 & 0 & 0 & 0 & 0 & 1 & 2 & 5 & 2 & 1 & 6 \\
0 & 1 & 0 & 0 & 0 & 3 & 7 & 8 & 2 & 5 & 1 \\
0 & 0 & 1 & 0 & 0 & 8 & 1 & 2 & 6 & 1 & 8 \\
0 & 0 & 0 & 1 & 0 & 1 & 1 & 6 & 4 & 7 & 7 \\
0 & 0 & 0 & 0 & 1 & 2 & 5 & 2 & 1 & 6 & 1 \\
\end{array}
\right),
\end{equation*}
\begin{equation*}
G_{12}=G_{21}=\left(
\begin{array}{ccccccccccc}
1 & 0 & 0 & 0 & 0 & 1 & 6 & 1 & 2 & 5 & 2 \\
0 & 1 & 0 & 0 & 0 & 7 & 7 & 4 & 6 & 1 & 1 \\
0 & 0 & 1 & 0 & 0 & 8 & 1 & 6 & 2 & 1 & 8 \\
0 & 0 & 0 & 1 & 0 & 1 & 5 & 2 & 8 & 7 & 3 \\
0 & 0 & 0 & 0 & 1 & 6 & 1 & 2 & 5 & 2 & 1 \\
\end{array}
\right).
\end{equation*}
Let $\widehat{\mathcal {D}}_1$ and $\widehat{\mathcal {D}}_2$ be the extensions of $\mathcal {D}_1$ and $\mathcal {D}_2$ respectively. Then, by Lemma 16, we have that $\widehat{\mathcal {D}}_1$ and $\widehat{\mathcal {D}}_2$ have the following matrices as their generator matrices respectively
$$\widehat{G}_1=\left(\begin{array}{cccc}8&5&\cdots&5\\ 0& & &\\\vdots& & \huge{G}_1 &\\0& & & \end{array}\right),~~~~\widehat{G}_2=\left(\begin{array}{cccc}8&5&\cdots&5\\ 0& & &\\\vdots& & \huge{G}_2 &\\0& & & \end{array}\right).$$
The codes $\widehat{\mathcal {D}}_1$ and $\widehat{\mathcal {D}}_2$ are Euclidean self-dual of lengths $12$ over $R$. Furthermore, they are equivalent to each other. From Proposition 2, $\Phi(\widehat{\mathcal {D}}_1)$ and $\Phi(\widehat{\mathcal {D}}_2)$ are $\mathbb{Z}_9$-Euclidean self-dual codes of length $24$ with minimum Hamming weight $9$. The Hamming weight distribution of $\Phi(\widehat{\mathcal {D}}_1)$ is \\
$1+5632x^{9}+63360x^{11}+720912x^{12}+4580928x^{13}+30739104x^{14}+164535360x^{15}+730121040x^{16}
+2756179008x^{17}+8597448640x^{18}+21680524800x^{19}+43367140080x^{20}+66118443072x^{21}
+72092601504x^{22}+50166642240x^{23}+16719790800x^{24}.$
\par
Since the number of codewords in $\Phi(\widehat{\mathcal {D}}_1)$ equals $9^{12}$, the Singleton bound for the
minimal distance $d$ of $\Phi(\widehat{\mathcal {D}}_1)$  is $d\leq 24-\log_{9}^{|\Phi(\widehat{\mathcal {D}}_1)|}+1=13$.
Let $w(\Phi(\widehat{\mathcal {D}}_1),\leq 12)$ be the set of codewords
in $\Phi(\widehat{\mathcal {D}}_1)$ having Hamming weight at most $12$.
Then a direct calculation shows that
$\frac{|w(\Phi(\widehat{\mathcal {D}}_1),\leq 12)|}{|\Phi(\widehat{\mathcal {D}}_1)|}\leq 2.8\times 10^{-6}$.
\end{example}

\vskip 3mm \noindent {\bf Acknowledgments}  \emph{This research is supported by the National Key Basic Research Program of China (Grant No. 2013CB834204), and the National Natural Science Foundation of China (Grant No. 61171082).}






\end{document}